\documentclass[12pt]{amsart}

\usepackage{amsmath, amssymb, graphics}

\newcommand{\mathsym}[1]{{}}
\newcommand{\unicode}[1]{{}}

\newtheorem{thm}{Theorem}[section]
\newtheorem{cor}[thm]{Corollary}
\newtheorem{lem}[thm]{Lemma}
\newtheorem{prop}[thm]{Proposition}
\theoremstyle{definition}
\newtheorem{notation}[thm]{Notation}
\newtheorem{defn}[thm]{Definition}
\newtheorem{rem}[thm]{Remark}

\newtheorem*{defn*}{Definition}
\newtheorem*{rems*}{Remarks}
\newtheorem*{rem*}{Remark}

\numberwithin{equation}{section}
\begin{document}

\title [The Wigner caustic on shell ] {The Wigner caustic on shell and \\ singularities of odd functions}

\author[Domitrz, Manoel, Rios]{Wojciech Domitrz, Miriam Manoel and Pedro de M Rios}
\address{Warsaw University of Technology, Faculty of Mathematics and Information Science, Plac Politechniki 1, 00-661 Warszawa, Poland}
\email{domitrz@mini.pw.edu.pl}
\address{Departamento de Matem\'atica, ICMC, Universidade de S\~ao Paulo; S\~ao Carlos, SP, 13560-970, Brazil}
\email{miriam@icmc.usp.br}
\address{Departamento de Matem\'atica, ICMC, Universidade de S\~ao Paulo; S\~ao Carlos, SP, 13560-970, Brazil}
\email{prios@icmc.usp.br}

\thanks{W. Domitrz was supported by FAPESP/Brazil and by Polish MNiSW grant no. N N201 397237 during his visits to ICMC-USP, S\~ao Carlos. P. de M. Rios received partial support by
FAPESP/Brazil  for his visits to Warsaw.}

\subjclass{58K40, 53D12, 81Q20, 58K70, 58K50.}

\keywords{Semiclassical dynamics, Symplectic geometry, Lagrangian singularities, Simple singularities, Symmetric singularities}

\maketitle

\begin{abstract} We study  the Wigner caustic on shell of a Lagrangian submanifold $L$ of affine symplectic space. We present the physical motivation for studying singularities of the Wigner caustic on shell and present its mathematical definition in terms of a generating family. Because such a generating family is an odd deformation of an odd function, we study simple singularities in the category of odd functions and their odd versal deformations, applying these results to classify the singularities of the Wigner caustic on shell,  interpreting these singularities in terms of the local geometry of $L$.

\end{abstract}

\section{Introduction}

The Wigner caustic of a smooth convex closed curve $L$ on  affine symplectic plane was first introduced by Berry, in his celebrated 1977 paper \cite{Ber} on the semiclassical limit of Wigner's phase-space representation of quantum states. Thus, when $L$ is the classical correspondence of a pure quantum state, the Wigner function of this state takes on high values, in the semiclassical limit, at points in a neighborhood of $L$ and also in a neighborhood of  a singular closed curve in its interior, generically formed by an odd number of cusps: the Wigner caustic of $L$.

Some years later, Ozorio de Almeida and Hannay studied  the Wigner caustic of a smooth Lagrangian torus $L$ on affine symplectic $4$-space \cite{OH}. Since their main object of study was the geometrical place where the amplitude of the Wigner function of the pure quantum state corresponding to $L$ rises considerably, in the semiclassical limit, they considered $L$ itself as part of the Wigner caustic and focused some attention on the part of the Wigner caustic that is close to and contains $L$.

From a purely geometrical point of view, the Wigner caustic of $L$, hereby denoted ${\bf E}_{{1}/{2}}(L)$, is defined as the locus of midpoints of segments
connecting pairs of points on $L$ with  ``parallel'' affine tangent spaces.
Here, parallelism is taken in a broad sense, also allowing for {\it weak} parallelism, when the direct sum of the tangent spaces of $L$ at the two points do not span the whole  $\mathbb R^{2m}$.
However, as mentioned above, from the perspective of applications of Wigner caustics in quantum physics, it is interesting to consider
an even broader definition of parallelism, when a single point of $L$ is identified as a pair of points with parallel affine tangent spaces (in this case {\it strongly} parallel spaces). Then, with this extended notion in the geometrical definition,
the submanifold $L$ itself is a subset of  ${\bf E}_{{1}/{2}}(L)$. The part of  ${\bf E}_{{1}/{2}}(L)$ that is close to $L$ and that contains $L$ is called the {\it Wigner caustic on shell}.

In this paper, we study the Wigner caustic on shell of a smooth Lagrangian submanifold $L$ of the affine symplectic space $(\mathbb R^{2m}, \omega)$, focusing on its Lagrangian-stable singularities when $L$ is a curve or a surface. Its definition in terms of a generating family reveals the fact that the Wigner caustic on shell has a (hidden) symmetry under the action of $\mathbb Z_2$, because its generating family is an odd deformation of an odd function of the variables. No such symmetry exists for the part of the Wigner caustic that is away from $L$, whose  simple stable Lagrangian singularities have been studied in a previous paper \cite{DRs}.

Now, our interest in studying singularities of the Wigner caustic stems from semiclassical dynamics. Because the amplitude of the Wigner function rises sharply along the Wigner caustic, in the semiclassical limit, there is where uniform asymptotic expressions must be used. However, the kind of uniform asymptotic expression for the semiclassical Wigner function in a neighborhood of a point varies according to the kind of singularity of the Wigner caustic at that point \cite{Ber}. Thus, for a finer treatment of  the dynamics of the semiclassical Wigner function of a pure quantum state \cite{RO1}, it is important to classify the singularities of the Wigner caustic (off and on shell) of a Lagrangian submanifold,  which are stable under the group of symplectomorphisms  of
$(\mathbb R^{2m}, \omega)$.

Because such singularities are described by generating families, here we focus attention on simple singularities of function-germs (simple here in the classical notion of
absence of modal parameters \cite{AGV}) and their versal deformations. Thus, for the Wigner caustic on shell,
our first aim is to obtain the list of all simple singularities  in the category of odd-functions. This paper is, therefore, divided in three parts.

The first part, Section 2, presents the motivation and definition
of the Wigner caustic on shell of a Lagrangian submanifold.

The second part, Section 3, is independent of the other sections and  is devoted to the
classification of simple singularities  of odd functions and their odd deformations.  By odd function-germs at $0\in\mathbb R^m$ we mean $\mathbb Z_2$-equivariant smooth function-germs, with $\mathbb Z_2$ action on the source: $(x_1,\cdots,x_m) \mapsto (-x_1,\cdots, -x_m)$ and on the target: $y \mapsto -y.$ We classify odd function-germs using classical $\mathcal R$-equivalence (composition with germs of diffeomorphisms on the source) restricted to the  subgroup of odd diffeomorphism-germs, which is natural in this context. We prove  there are no simple odd singularities if the dimension of the source is greater than two and classify all simple odd function-germs in dimensions one and two, presenting their odd mini-versal deformations. Although this could be considered as a classical subject in singularity theory, surprisingly no such classification list of simple odd singularities has been found by the authors in the literature.

In one variable the simple odd singularities are of type
that we shall denote $A_{2k/2}$, which have codimension $k$ in the category of odd function-germs and which coincide with an intersection of the classical $\mathcal R$-orbit of $A_{2k}$ singularities of codimension $2k$ with the module of odd function-germs. In two variables, the simple odd singularities are divided in two groups: the first one  of types hereby denoted $D^{\pm}_{2k/2}$ and $E_{8/2}$, of  odd codimensions $k$ and $4$ respectively, which are the intersections of classical $\mathcal R$-orbits  of types  $D^{\pm}_{2k}$ and $E_{8}$, of codimensions $2k$ and $8$ respectively, with the module of odd function-germs. The second group consists of the singularities of types hereby denoted  $J_{10/2}^{\pm}$ and $E_{12/2}$, of respective odd codimensions $5$ and $6$, these notations chosen because they are $\mathcal R$-equivalent to singularities  $J_{10}$ and $E_{12}$ of respective codimensions $10$ and $12$,  these later being unimodal in Arnold's classification.

The third part, Section 4, applies the results of Section 3. For Lagrangian curves, we give the conditions for realizing the odd deformations of  singularities  $A_{2/2}$ and $A_{4/2}$ as generating families for simple stable Lagrangian singularities of the Wigner caustic on shell, and describe these singularities. For Lagrangian surfaces, we present the realization conditions for  the singularities of the Wigner caustic on shell of  types $D^{\pm}_{2k/2}$, $k=2,3,4$, and $E_{8/2}$. Because  the odd codimension in this context can be at most $4$, these are all the simple singularities that can be realized as  simple stable Lagrangian singularities of the Wigner caustic on shell.
Finally, we also interpret the realization condition of each of these singularities of the Wigner caustic on shell in terms of the local geometry of the Lagrangian curve or  the Lagrangian surface.

While working on this paper, we benefitted from discussions with F. Tari and specially with M. A. S. Ruas, to whom both we are grateful.

\section{The Wigner caustic on shell}

\subsection{Physical origins of the Wigner caustic on shell} \label{subseq:physical} The following presentation is sketchy and  can be found expanded in various textbooks and research papers (see \cite{Ber, OH, RO1}, for instance).

We recall that, in non-relativistic quantum mechanics, a {\it pure state of the system} is usually defined as a normalized vector $\Psi$ in a Hilbert space $\mathcal H$. In many simple cases, ${\mathcal H}= L^2_{\mathbb C}(\mathbb R^m)$, the space of complex-valued square-integrable functions on $\mathbb R^m$. Here,  $\mathbb R^m$ is commonly interpreted either as the {\it configuration-space} $Q$ or the {\it momentum-space} $P$ and $m\in \mathbb N$ is the number of {\it degrees of freedom} of the system.

The Fourier transform $\mathcal F: L_{\mathbb C}^2(\mathbb R^m) \to L_{\mathbb C}^2(\mathbb R^m)$ relates configuration-space and momentum-space representations of a state $\Psi$, by
$$\psi(q)\mapsto {\mathcal F}_{\psi}(p)=\frac{1}{(2\pi\hbar)^{m}}\int_{\mathbb R^m} \psi(q)\exp{(ipq/\hbar)} \ dq \ ,$$
where $i=\sqrt{-1}$ and $\hbar$ is a positive constant, called Planck's constant, which provides a scale for comparing quantum to classical phenomena.

On the other hand, in classical conservative dynamics, the concept of a {\it phase-space} $\Pi$ is predominant. In the simple cases when $Q=P=\mathbb R^m$,  $\Pi=P\times Q=\mathbb R^{2m}$, endowed with the symplectic form $\omega=\sum_{i=1}^m dp_i
\wedge dq_i$, is an affine-symplectic space.

The Wigner transform ${\mathcal W}: L_{\mathbb C}^2(\mathbb R^m)\to L^1_{\mathbb R}(\mathbb R^{2m}, \omega)$ defines a phase-space representation of a pure state $\Psi$, called its {\it Wigner function}, from the configuration-space representation of $\Psi$, by $$\psi(q)\mapsto {\mathcal W}_{\psi}(p,q)=\frac{1}{(\pi\hbar)^m}\int_{\mathbb R^m} \psi^*(q-\zeta)\psi(q+\zeta)\exp{(2ip\zeta/\hbar)} \ d\zeta  \ .$$
The Wigner function satisfies reality and Liouville-normalization,
$${\mathcal W}_{\psi}(p,q)={\mathcal W}^*_{\psi}(p,q) \ , \ \int_{\mathbb R^{2m}} {\mathcal W}_{\psi}(p,q) dpdq = 1 \ , \ dpdq=\omega^m/m! $$
and, although ${\mathcal W}_{\psi}(p,q)$ can be negative, its partial integrals are not,
$$\int_{\mathbb R^m} {\mathcal W}_{\psi}(p,q) dp = |\psi(q)|^2\geq 0 \ , \ \int_{\mathbb R^m} {\mathcal W}_{\psi}(p,q) dq = |{\mathcal F}_{\psi}(p)|^2\geq 0 \ ,$$
so that ${\mathcal W}_{\psi}$ can be seen as a pseudo probability distribution on phase-space $(\mathbb R^{2m}, \omega)$, while $|\psi|^2$ and $ |{\mathcal F}_{\psi}|^2$ are actual probability distributions on configuration-space and momentum-space, respectively.

In various instances, one is mostly interested in a pure state $\Psi$ which is eigenstate of one or more self-adjoint operators on $\mathcal H = L^2_{\mathbb C}(\mathbb R^m)$. If $F\in {\mathcal B}(\mathcal H)$ is self-adjoint, its classical correspondence is a real function $f\in \mathcal C^{\infty}_{\mathbb R}(\mathbb R^{2m}, \omega)$ so that, if $F(\Psi)=\alpha\Psi, \alpha\in\mathbb R$, then $\Psi$ corresponds classically to the level set $\Lambda=\{x=(p,q)\in \mathbb R^{2m}:f(x)=\alpha\}$, which for many values of $\alpha$ is a smooth hypersurface in phase-space (a smooth Lagrangian curve $\Lambda =L$ for systems with one degree of freedom).

For systems with $m>1$ degrees of freedom, two linearly independent functions  $f_1, f_2\in   \mathcal C^{\infty}_{\mathbb R}(\mathbb R^{2m}, \omega)$ are said to be in {\it involution} if $X_{f_1}(f_2)=X_{f_2}(f_1)=0$, where $X_{f_j}$ is the vector field defined by Hamilton's equation $df_j+X_{f_j}\lrcorner\omega=0$. If there exist $m$ linearly independent functions $f_j$ in mutual involution, the classic dynamical system is integrable and each level set $L=\{x\in \mathbb R^{2m}:f_j(x)=\alpha_j \in\mathbb R, j=1,\dots, m \}$ is a Lagrangian submanifold of $(\mathbb R^{2m}, \omega)$. Such $L$ may correspond to a pure state $\Psi$ which is eigenstate of $m$ linearly independent commuting self-adjoint operators  $F_j\in {\mathcal B}(\mathcal H)$, $[F_i,F_j]=0$, $F_j(\Psi)=\alpha_j\Psi$.

The semiclassical approximation of $\Psi$ can be formally seen as the asymptotic expansion on $\hbar << 1$ of some representation of $\Psi$. Let's start with the  crude expression for the semiclassical approximation  of the
Wigner function of a pure state in one  degree of freedom \cite{Ber}:
\begin{equation}\label{WF1}
{\mathcal W}_{\psi}(x)\approx \sum_k{\mathcal A}_k^{\hbar}(x)\cos{(S_k(x)/\hbar -\pi/4)} \ ,
\end{equation}
where $S_k(x)$ is the symplectic area enclosed by the curve $L=\{x'\in \mathbb R^{2}:f(x')=\alpha\}$ and the $k$-th chord connecting two points $x^+_k$ and $x^-_k$ on $L$, whose midpoint is $x$ (for $x$ close to $L$, such a chord is often unique, or does not exist). Each amplitude function
${\mathcal A}_k^{\hbar}(x)$ in (\ref{WF1}) satisfies
\begin{equation}\label{AWF}
{\mathcal A}_k^{\hbar}(x)\propto \frac{1}{|\omega(X_f^{+k}(x),X_f^{-k}(x))|^{1/2}} \ ,
\end{equation}
where $X_f^{\pm k}(x)$ is the Hamiltonian vector field $X_f$ evaluated at the endpoint $x^{\pm}_k\in L$ of the $k$-th chord, parallel translated to its centre $x$.

The number of chords centered on $x$ connecting pairs of points on $L$ varies, as $x$ varies, and its bifurcation set is given by
\begin{equation}\label{WCWF}
{\bf E}_{{1}/{2}}(L)= \{x\in \mathbb R^{2} \ : \ \exists k \ \ \omega(X_f^{+k}(x),X_f^{-k}(x))=0\} \ .
\end{equation}

 It is clear from (\ref{WCWF}) that
${\bf E}_{{1}/{2}}(L)$ can be defined as the set of midpoints of chords connecting points on $L$ whose tangent vectors to $L$ at these endpoints are parallel.
${\bf E}_{{1}/{2}}(L)$ is called the {\it Wigner caustic} of $L$ and is precisely the set where some ${\mathcal A}_k^{\hbar}$ blows up to infinity, see (\ref{AWF}).

In fact, in a neighborhood of ${\bf E}_{{1}/{2}}(L)$, the crude expression (\ref{WF1}) is inappropriate and must be substituted by uniform approximations that do not blow up to infinity on ${\bf E}_{{1}/{2}}(L)$ if $\hbar\neq 0$ but, nonetheless, take on very high values at ${\bf E}_{{1}/{2}}(L)$ for $\hbar << 1$. However, the kind of uniform approximation to be used will depend on the kind of singularity of the Wigner caustic. Thus, where the Wigner caustic corresponds to a fold singularity, the uniform approximation of the Wigner function is written in terms of Airy functions but, where the Wigner caustic has cusp singularities, Pearcey functions must be used (see \cite{Ber}).

Now, it is obvious from (\ref{WCWF}) that $L\subset {\bf E}_{{1}/{2}}(L)$, so that ${\mathcal W}_{\psi}$ peaks at $L$ for $\hbar << 1$. On the other hand, as $x\to L$, $S(x)\to 0$ and $\nabla S(x)\to 0$, so that  ${\mathcal W}_{\psi}$  is not highly oscillatory in a small neighborhood of $L$,  for $\hbar << 1$. This contrasts sharply with the situation when $x$ is far from $L$ where, even if $x\in  {\bf E}_{{1}/{2}}(L)$,  ${\mathcal W}_{\psi}$  is highly oscillatory for $\hbar << 1$ and tends on average to $0$ in any small neighborhood of $x$, as $\hbar\to 0$. Thus, as $\hbar\to 0$, the pseudo probability distribution  ${\mathcal W}_{\psi}$ tends on average to the singular probability distribution which is zero everywhere but on $L$, where ${\mathcal W}_{\psi}$ tends to infinity. In this way, $L$ can be seen as the classical correspondence of the pure state $\Psi$.

The less oscillatory behavior of the Wigner function ${\mathcal W}_{\psi}$ in a neighborhood of $L$ makes it convenient to separate the Wigner caustic of $L$  in a part which is away from $L$ and another which is very close to $L$ and contains $L$. This latter  is called the {\it Wigner caustic on shell}.

The situation for integrable systems with more degrees of freedom is similar: the crude semiclassical expression for  the Wigner function is
\begin{equation}\label{WFm}
{\mathcal W}_{\psi}(x)\approx \sum_k{\tilde{\mathcal A}}_k^{\hbar}(x)\cos{(\tilde{S}_k(x)/\hbar -n_k\pi/4)} \ ,
\end{equation}
where $\tilde{S}_k(x)$ is the symplectic area of any surface bounded by a curve formed by taking any arc of the Lagrangian submanifold $L=\{x'\in\mathbb R^{2m}:f_j(x')=\alpha_j, j=1,...,m\}$ and closing it with the $k$-th chord connecting two points $x^+_k$ and $x^-_k$ on $L$, with midpoint $x$, and where
\begin{equation}\label{AWFm}
\tilde{{\mathcal A}}_k^{\hbar}(x)\propto \frac{1}{|\det[\omega(X_{f_i}^{+k}(x),X_{f_j}^{-k}(x))]|^{1/2}} \ ,
\end{equation}
with $X_{f_j}^{\pm k}(x)$ being the Hamiltonian vector field $X_{f_j}$ evaluated at the endpoint $x^{\pm}_k\in L$ of the $k$-th chord, parallel translated to its centre $x$.  Also, the integer $n_k$ in (\ref{WFm}) is  the signature of the $m\times m$ matrix
$[\omega(X_{f_i}^{+k}(x),X_{f_j}^{-k}(x))]$.
Therefore, in this case,
\begin{equation}\label{WCWFm}
{\bf E}_{{1}/{2}}(L)= \{x\in \mathbb R^{2m} \ : \ \exists k \ \ \det[\omega(X_{f_i}^{+k}(x),X_{f_j}^{-k}(x))]=0\} \
\end{equation}
and  can be identified with the set  of midpoints of chords connecting points on $L$ whose tangent spaces to $L$ at these endpoints are {\it weakly parallel}, in other words, do not span the whole $\mathbb R^{2m}$, see \cite{OH}. Again, uniform approximations must be used instead of (\ref{WFm}) in a neighborhood of ${\bf E}_{{1}/{2}}(L)$ and,
for $\hbar << 1$, ${\mathcal W}_{\psi}$ is not highly oscillatory in a small neighborhood of $L$, which is the classical correspondence of $\Psi$, and it is therefore natural to single out the Wigner caustic on shell.

\subsection{Mathematical definition of the Wigner caustic on shell} \label{subseq:mathematical}
 Let $L$ be a smooth Lagrangian submanifold  of the
affine symplectic space $(\mathbb R^{2m},\omega=\sum_{i=1}^m dp_i
\wedge dq_i)$.   Let $a, b$ be points of $L$ and let
$\tau_{a-b}:\mathbb R^{2m} \ni x\mapsto x+(a-b) \in \mathbb R^{2m}$ be the
translation by the vector $(a-b)$.

\begin{defn}\label{parallelism} A pair of points $a, b \in L$ is a
{\bf weakly parallel} pair if
$$T_aL + \tau_{a-b}(T_bL)\ne \mathbb R^{2m}.$$

A weakly parallel pair $a, b \in L$ is
called {\bf $k$-parallel} if
$$\dim(T_aL \cap \tau_{b-a}(T_bL))=k.$$
If $k=m$ the pair $a, b \in L$ is called {\bf strongly parallel},
or just {\bf parallel}. \end{defn}
\begin{defn}
A {\bf chord} passing through a pair $a,b$, is the line
$$
l(a,b)=\{x\in \mathbb R^n:x=\eta a + (1-\eta) b, \eta \in
\mathbb R\}.
$$
\end{defn}
\begin{defn} For a given $\eta$, an {\bf affine
$\eta$-equidistant} of $L$,  denoted
${\bf E}_{\eta}(L)$, is the set of all
$x\in \mathbb R^{2m}$ s.t. $x=\eta a + (1-\eta) b$, for
all weakly parallel pairs $a,b \in L$.
Note that, for any
$\eta$, ${\bf E}_{\eta}(L)={\bf E}_{1-\eta}(L)$ and in particular
${\bf E}_0(L)={\bf E}_1(L)=L$. Thus, the case $\eta=1/2$ is special.
\end{defn}
\begin{defn} The set ${\bf E}_{{1}/{2}}(L)$ is the {\bf Wigner caustic} of $L$.
\end{defn}

Consider  $\mathbb R^{2m}\times \mathbb R^{2m}$ with coordinates $(x^+,x^-)$
and the tangent bundle to $\mathbb R^{2m}$,  $T\mathbb
R^{2m}=\mathbb R^{2m}\times \mathbb R^{2m}$, with coordinates
$(x,\dot{x})$ and standard projection
$\pi: T\mathbb R^{2m}\ni (x,\dot{x})\rightarrow x \in \mathbb R^{2m}$.
Consider the linear map
$$\Phi_{1/2}:\mathbb R^{2m}\times \mathbb
R^{2m}\ni(x^+,x^-)\mapsto \left(\frac{x^+ + x^-}{2}, \ \frac{x^+ - x^-}{2}\right)=(x,\dot{x}) \in T\mathbb R^{2m}.$$

On the product affine symplectic space, consider the symplectic form
$$\delta_{1/2}\omega=\frac{1}{2}\left(
\pi_1^{\ast}\omega-\pi_2^{\ast}\omega\right) \ ,$$
$\pi_i$ the $i$-th projection $\mathbb R^{2m}\times \mathbb R^{2m}\to\mathbb R^{2m}$.
Canonical  relations correspond to Lagrangian submanifolds  of $(\mathbb R^{2m}\times \mathbb R^{2m}, \delta_{1/2}\omega)$. Then,
$$\left(\Phi_{1/2}^{-1}\right)^{\ast}(\delta_{1/2}\omega)
 \ = \ \dot{\omega} \ ,$$
where $\dot{\omega}$
is the canonical symplectic form on  $T\mathbb{R}^{2m}$, which is defined by $\dot{\omega}(x,\dot{x})=d\{\dot{x}\lrcorner\omega\}(x)$ or, in Darboux coordinates for $\omega$, by
$$\dot{\omega}=\sum_{i=1}^m
d\dot{p_i}\wedge dq_i+ dp_i\wedge d\dot{q_i} \ .$$

If $L$ is a Lagrangian submanifold of $(\mathbb{R}^{2m},\omega)$, then $L\times L$ is a Lagrangian submanifold of $(\mathbb{R}^{2m}\times\mathbb{R}^{2m},\delta_{1/2} \omega)$ and $\mathcal L= \Phi_{1/2}(L\times L)$ is a Lagrangian submanifold of $(T\mathbb{R}^{2m},\dot\omega)$, which  can be locally described by a generating function of the midpoints $x=\pi\circ\Phi_{1/2}(x^+,x^-)$, $(x^+,x^-)\in L\times L$, when  $\mathcal L$ projects regularly to the zero section \cite{Poi}\cite{RO2}.

We recall basic definitions of the theory of Lagrangian singularities (see \cite{AGV}, \cite{DRs}).
First, $(T\mathbb{R}^{2m},\dot\omega)$ with canonical projection $\pi: T\mathbb{R}^{2m}\rightarrow \mathbb R^{2m}$ is a {\it Lagrangian fibre bundle} and $\pi|_{\mathcal L}:\mathcal L\rightarrow \mathbb R^{2m}$ is a {\it Lagrangian map} .  Let $\tilde {\mathcal L}$ be another Lagrangian submanifold of $(T\mathbb{R}^{2m},\dot\omega)$. Two Lagrangian maps $\pi|_{\mathcal L}:\mathcal L\rightarrow \mathbb R^{2m}$ and $\pi|_{\tilde{\mathcal L}}:\tilde {\mathcal L}\rightarrow \mathbb R^{2m}$ are {\it Lagrangian equivalent} if there exists a symplectomorphism of   $(T\mathbb{R}^{2m},\dot\omega)$ taking  fibres of  $\pi$  to fibres and mapping $\mathcal L$ to $\tilde{\mathcal L}$. A Lagrangian map is {\it stable} if every nearby Lagrangian map (in the Whitney topology) is Lagrangian equivalent to it. The
set of critical values of a Lagrangian map is called a {\it caustic}.
Then, we have the following result:
\begin{prop}[\cite{DRs}]\label{caustic}
The caustic of the Lagrangian map $\pi|_{\mathcal L}:\mathcal L\rightarrow \mathbb R^{2m}$ is the Wigner caustic ${\bf E}_{{1}/{2}}(L)$.
\end{prop}

In this paper, we study
${\bf E}_{{1}/{2}}(L)$ in a neighborhood $L$. For this reason, we consider pairs of points of the type $(a,a)\in L\times L$ as strongly parallel pairs. In other words, in Definition \ref{parallelism} {\it we did not impose the restriction $a\neq b$ on the pair of points of $L$ to be considered a parallel pair}. This broader definition of parallel pairs is suitable for studying the part of the Wigner caustic that is close to $L$, because then $L$ is itself part of the Wigner caustic. This broader definition of the Wigner caustic is also natural from its origin in quantum physics, as shown by  equations (\ref{WCWF}) and (\ref{WCWFm}). On the other hand, imposing the restriction $a\neq b$ in Definition \ref{parallelism} allows for a neater definition of the Wigner caustic as a centre symmetry set, as in \cite{DRs} (see also \cite{Gib2}, where, for a curve $L$ and $a\neq b$, ${\bf E}_{{1}/{2}}(L)$ is called the area evolute of $L$).
\begin{defn}
The germ at $a$ of the {\bf Wigner caustic on shell} is the germ of Wigner caustic ${\bf E}_{{1}/{2}}(L)$ at the point $a\in L$.
\end{defn}

Now let $L$ be a germ at $0$ of a smooth Lagrangian submanifold of $(\mathbb{R}^{2m},\omega)$, generated by the function-germ $S\in \mathcal
E_{m}$ in the usual way,
\begin{equation}\label{usual}
L=\left\{(p,q)\in \mathbb R^{2m}: p_i=\frac{\partial S}{\partial
q_i}(q) \ \text{for} \ i=1,\cdots,m \right\} \ .
\end{equation}
Then, $\mathcal L$ is the germ at $0$ of a submanifold of $(T\mathbb R^{2m}, \dot\omega)$
described as
\begin{equation}
 \dot{p}=\frac{1}{2}\left(\frac{\partial S}{\partial
q}(q+\dot{q})-\frac{\partial S}{\partial
q}(q-\dot{q})\right),\label{L1}
\end{equation}
\begin{equation}
p=\frac{1}{2}\left(\frac{\partial S}{\partial
q}(q+\dot{q})+\frac{\partial S}{\partial
q}(q-\dot{q})\right).\label{L2}
\end{equation}
By Proposition \ref{caustic}, the germ at $0\in L$ of the Wigner caustic on shell ${\bf E}_{{1}/{2}}(L)$ is described as
\begin{eqnarray}
\exists \dot{q}\in \mathbb R^m \  s.t.  \ (\ref{L2}) \ \ \mbox{is satisfied, and} \nonumber
\\
\det\left[\frac{\partial^2 S}{\partial q^2}(q+\dot{q})-\frac{\partial^2 S}{\partial q^2}(q-\dot{q})\right]=0. \label{WC2}
\end{eqnarray}

Thus, putting $\dot{q}=0$ in (\ref{L2})-(\ref{WC2}) we obtain the obvious fact:

\begin{prop}
$L$ is contained in ${\bf E}_{{1}/{2}}(L)$.
\end{prop}

Now, let us consider the reflection
\begin{equation}\label{zeta}\zeta:T\mathbb R^{2m}\ni (\dot{p},\dot{q},p,q)\mapsto
(-\dot{p},-\dot{q},p,q)\in T\mathbb R^{2m}\end{equation} whose mirror is the
zero section $\{\dot{p}=\dot{q}=0\}\subset T\mathbb R^{2m}$. Note that $\{ id, \zeta \}$ generates an action of $\mathbb Z_2$ on $T\mathbb R^{2m}$.
Using
(\ref{L1}) we obtain
\begin{prop}\label{mirror}
$\mathcal L$ is $\mathbb Z_2$-symmetric, that is,
$\zeta(\mathcal L)=\mathcal L$. \end{prop}

We shall study singularities of  ${\bf E}_{{1}/{2}}(L)$ via generating families of $\mathcal L$.
\begin{defn}
The germ of a {\bf generating family} of $\mathcal L$ is the smooth function-germ  $F:\mathbb R^k\times \mathbb R^{2m}\ni (\beta,p,q)\mapsto F(\beta,p,q)\in \mathbb R$ such that
\begin{equation}\label{genL}
\mathcal L=\left\{(\dot{p},\dot{q},p,q)\in T\mathbb
R^{2m}: \ \exists \ \beta \in \mathbb R^k \ \  \dot{p}=\frac{\partial F}{\partial q}, \ \dot{q}=-\frac{\partial F}{\partial p}, \ \frac{\partial F}{\partial \beta}=0\right\}.
\end{equation}
\end{defn}
\begin{rem}\label{Lag-stable}When there are no symmetries, two Lagrangian map-germs on the same Lagrangian fibre bundle are Lagrangian equivalent if and only if their generating families are stably (fibred) $\mathcal R^+$-equivalent. Moreover the Lagrangian map-germ given by the generating family $F(\beta,p,q)$ with parameters $(p,q)$ is Lagrangian stable if and only if $F(\beta,p,q)$ is a $\mathcal R^+$-versal deformation of $f(\beta)=F(\beta,0,0)$ (see \cite{AGV}).\end{rem}

Now, in the $\mathbb Z_2$-symmetric context,  the following Theorem, whose proof is a straightforward computation from (\ref{genL}) to (\ref{L1})-(\ref{L2}), is a particular case of the more general result presented in \cite{DRs}:
\begin{thm}[\cite{DRs}]\label{genfam}
The germ at $0\in L$ of the Wigner caustic on shell
is the germ of a caustic of the germ of a Lagrangian
submanifold $\mathcal L$ in the Lagrangian fibre bundle $T\mathbb
R^{2m}\ni (\dot{p},\dot{q},p,q)\mapsto (p,q)\in \mathbb R^{2m}$
with the symplectic form $\dot{\omega}=\sum_{i=1}^m
d\dot{p}_i\wedge dq_i+dp_i \wedge d\dot{q}_i$ and
generating family
\begin{equation}\label{gf}
F(\beta,p,q)\equiv
\frac{1}{2}S(q+\beta)-\frac{1}{2}S(q-\beta)-\sum_{i=1}^m p_i
\beta_i.
\end{equation}
\end{thm}
For any $\beta$, $p$, $q$, the generating family (\ref{gf}) satisfies
\begin{equation}\label{oddF}
F(-\beta,p,q)\equiv-F(\beta,p,q)
\end{equation}
It implies that $F$ is a deformation of an {\it odd function-germ}
\begin{equation}\label{f}
f(\beta)\equiv F(\beta,0,0)\equiv \frac{1}{2}(S(\beta)-S(-\beta)).
\end{equation}
We call $F$ which satisfies (\ref{oddF}) an {\it odd deformation} of an odd function-germ $f$ (see Definitions \ref{oddgerm} and \ref{odddeform}, below).
Thus, in order to study singularities of the Wigner caustic on shell,
we must consider classification  of odd function-germs and their odd deformations.
\begin{rem}\label{hidden} Theorem \ref{genfam}
implies that singularities of the Wigner caustic on shell are $\mathbb Z_2$-symmetric singularities (see Proposition \ref{mirror}, above, and Remark~\ref{remark equivariance}, below). However, at the level of a germ of the Wigner caustic on shell ${\bf E}_{{1}/{2}}(L)\subset\mathbb R^{2m}$, this $\mathbb Z_2$-symmetry is a {\it hidden symmetry} which is only actually revealed  in $\mathcal L \subset T\mathbb R^{2m}$. \end{rem}
\begin{rem} The form (\ref{gf}) for the generating family of the Wigner caustic on shell of a Lagrangian submanifold of the affine-symplectic space was already presented in \cite{OH}, and its odd character was remarked. However, the classification used there, borrowed from Arnold's, was not performed in the $\mathbb Z_2$-symmetric context. Furthermore, albeit respecting that $f(\beta)=F(\beta,0,0)$ is odd,
the authors did not take into account that $F(\beta,p,q)$ must be an odd deformation of $F(\beta,0,0)$. \end{rem}

\section{Singularities of odd functions}

\subsection{Preliminaries} \label{subsec:machinery}

We recall basic definitions.
\begin{defn}\label{oddgerm}
A smooth function-germ $f$ at $0$ on $\mathbb R^m$ is {\bf even} if $f(-x)\equiv f(x)$ and it is
 is {\bf odd} if $f(-x)\equiv -f(x)$.
\end{defn}
\begin{notation} Let us denote by $\mathcal E_m^{even}$ the ring of even smooth
function-germs $f : (\mathbb R^m, 0) \to {\mathbb R}$ and by
$\mathcal E_m^{odd}$ the set of odd smooth
function-germs $g : (\mathbb R^m, 0) \to ({\mathbb R},0)$, which has a
module structure over $\mathcal E_m^{even}$. \end{notation}

\begin{rem} \label{remark equivariance}
Consider the diagonal action of $\mathbb Z_2=\{1,-1\}$
on ${\mathcal R}^m$:
\begin{equation}\label{action}
\begin{array}{ccc}
\mathbb Z_2\times \mathbb R^m & \to &  \mathbb R^m \\
\bigl(\gamma, (x_1, \ldots, x_m)\bigr) & \mapsto & (\gamma x_1, \ldots, \gamma x_m).
\end{array}
\end{equation}
Hence, $\mathcal E_m^{even}$ is the ring of $\mathbb Z_2$-invariant germs under this action on
source. Also, $\mathcal E_m^{odd}$
is the module of $\mathbb Z_2$-equivariant germs, with same action on source
and on target - take (\ref{action}) for $m=1$.
\end{rem}

We now set up the equivalence relation  in $\mathcal E_m^{odd}$.
Changes of coordinates
shall preserve $\mathbb Z_2$-equivariance, so we  consider  the following:
\begin{defn}  A diffeomorphism-germ $\Phi: (\mathbb R^m,0) \rightarrow (\mathbb R^m,0)$ is {\bf odd} if  $\Phi(-x)\equiv -\Phi(x)$. Denote by $\mathcal D_m^{odd}$ the group of odd diffeomorphism-germs $(\mathbb R^m,0) \rightarrow (\mathbb R^m,0)$.
\end{defn}

\begin{defn}
Let $f, g\in\mathcal E_m^{odd}$. We say that $f$ and $g$ are {\bf $\mathcal R^{odd}$-equivalent} if
there exists $\Phi \in \mathcal D_m^{odd}$ such that $f=g\circ \Phi$.
\end{defn}

Following standard notation, denote by  $L \mathcal R^{odd} g$  the
tangent space to the $\mathcal R^{odd}$-orbit of $g$ at $g$, given by elements of the form
$\frac{d}{dt}\left|_{t=0}\right.\left(g\circ
\Phi^t\right)=\sum_{i=1}^m \frac{\partial g}{\partial x_i}\frac{d
\phi^t_i}{dt}\left|_{t=0}\right.$,
where $g\circ \Phi^t$
is a path in the $\mathcal R^{odd}$-orbit of $g$, with $\Phi^t=(\phi_1^t,\cdots,\phi_m^t)$ in $\mathcal
D_m^{odd}$ such that $\Phi^0=I$.
Now,  $\phi_i^t=\sum_{j=1}^m x_j h_{ij}^t$, with $h_{ij}^t\in \mathcal E_m^{even}$, so that $\frac{d}{dt}\left|_{t=0}\right.\left(g\circ
\Phi^t\right)=\sum_{i,j=1}^m x_j\frac{\partial g}{\partial
x_i}\frac{d h^t_{ij}}{dt}\left|_{t=0}\right.$,
$i,j=1,\cdots,m$. Since $h_{ij}^t\in \mathcal E_m^{even}$, so does $\frac{d
h^t_{ij}}{dt}\left|_{t=0}\right.$.  We have:

\begin{prop}
Let $g\in \mathcal E_m^{odd}$. The tangent space $L \mathcal R^{odd}g$ to the $\mathcal
R^{odd}$-orbit of $g$ at $g$ is the $\mathcal E_m^{even}$-module generated by
$\left\{x_j\frac{\partial g}{\partial x_i}:i,j=1,\cdots,m\right\}$.
\end{prop}

\begin{defn}\label{odddeform}
A function-germ $F\in \mathcal E_{m+k}$ is an {\bf odd deformation} of $f\in \mathcal
E_m^{odd}$ if $F|_{\mathbb R^m\times \{0\}}=f$ and for any fixed
$\lambda \in \mathbb R^k$ the function-germ $F|_{\mathbb R^m\times
\{\lambda\}}\in \mathcal E_m^{odd}$. The space $\mathbb R^k$ is
called the {\bf base} of the odd deformation $F$ and  $k$ is its {\bf
dimension}.
\end{defn}

\begin{defn}
The odd deformation $F\in \mathcal E_{m+k}$ is {\bf $\mathcal
R^{odd}$-versal} if every odd deformation of $f$ is $\mathcal
R^{odd}$-isomorphic to one induced from $F$ i.e. any odd deformation
$G\in \mathcal E_{m+l}$ of $f$ is representable in the form
$$
G(x,\lambda) \equiv F(\Phi(x,\lambda),\Lambda(\lambda)),
$$
 $\Phi:(\mathbb R^{m+k},0) \to (\mathbb R^m,0)$, $\Lambda:(\mathbb
R^k,0)\rightarrow (\mathbb R^l,0)$ smooth map-germs s.t.
$$
\Phi|_{\mathbb R^m\times \{\lambda\}} \in \mathcal D_m^{odd}, \
\Phi(x,0)\equiv x.
$$
An $\mathcal R^{odd}$-versal deformation $F\in \mathcal E_{m+k}$ of
$f\in \mathcal E_m^{odd}$ is {\bf $\mathcal R^{odd}$-miniversal} if the
dimension of the base has  its least possible value. This minimum value being
 the (odd) codimension of $f$.
\end{defn}

The group $\mathcal D_m^{odd}$ is a geometric subgroup in the sense of Damon (see \cite{Damon2}). The following theorem is a particular case of  \cite[Theorem 3.7]{BM}:

\begin{thm} \label{inf-versal}
Let $g \in \mathcal E_m^{odd}$. Then

\noindent (a)  A $k$-parameter deformation $G$ of $g$ is $\mathcal R^{odd}$-versal if and only if
\[\mathcal E_m^{odd}=\mathcal E_m^{even} \left\{x_j\frac{\partial g}{\partial
x_i}:i,j=1,\cdots,m\right\}+\mathbb R \left\{\frac{\partial
G}{\partial \lambda_\ell}|_{\mathbb R^m\times
\{0\}}:\ell=1,\cdots,k\right\}. \]

\noindent (b) If $W \subset \mathcal E_m^{odd} $ is a finite dimensional vector space
such that $\mathcal E_m^{odd} = L \mathcal R^{odd} g \oplus W$, and if $h_1, \ldots, h_s\in \mathcal E_m^{odd} $ is a basis
for $W$, then  $G(x,\lambda) \equiv g(x)+\sum_{j=1}^{s} \lambda_j h_j(x)$ is a $\mathcal R^{odd}$-miniversal deformation of $g$.
\end{thm}

We introduce the equivalence relation between odd deformations.

\begin{defn} Odd deformations $F, G\in \mathcal E_{m+k}$ are {\bf fibred  $\mathcal R^{odd}$-equivalent} if
there exists a fibred diffeomorphism-germ $\Psi \in \mathcal
D_{m+k}$ s.t. $\Psi(x,\lambda)\equiv
(\Phi(x,\lambda),\Lambda(\lambda))$, $\Phi|_{\mathbb R^m\times
\{\lambda\}} \in \mathcal D_m^{odd}$, $\forall\lambda \in \mathbb
R^k$, and
$
F=G\circ \Psi.
$
\end{defn}

\begin{notation} Let $\mathcal M_m^{k({odd})}$ denote the $\mathcal E_m^{even}$-submodule of $\mathcal E_m^{odd}$ generated by $x_1^{k_1}\cdots x^{k_m}, \forall k_1,\cdots,k_m\geq 0$,  s.t. $k_1+\cdots+k_m=k$.\end{notation}
Obviously, these are nontrivial submodules precisely when $k$ is odd. It follows the finite determinacy result for our particular case (see \cite{Damon2}, \cite{Rob1}-\cite{Rob2}):

\begin{prop}
$g\in \mathcal E_m^{odd}$ is finitely $\mathcal R^{odd}$-determined if and only if $\mathcal M_m^{k({odd})} \subset L \mathcal R^{odd} g$ for some odd positive integer $k$ .
\end{prop}
\begin{thm}\label{simpleless3} Let $g\in \mathcal E_m^{odd}$ with a singular point at $0$. If  $m\ge 3$,  then $g$ is not $\mathcal R^{odd}$-simple.
\end{thm}
\begin{proof}
If $0$ is a singular point of $g$  then $g \in
 {\mathcal M}_m^{3(odd)}$. Dimension of the space of $3$-jets at $0$ of
singular odd function-germs is  ${(m+2)(m+1)m}/{6}$. We
act on this space with $GL(m)$, of dimension $m^2$. But, for $m\ge 3$,
 ${(m+2)(m+1)m}/{6}> m^2$.
\end{proof}
\begin{rem}\label{noquadratic} If $g\in\mathcal E_{2+n}^{odd}$, the usual procedure of adding quadratic forms in the remaining $n$ variables cannot be performed. \end{rem}

Thus, classification of  simple odd singularities must be performed only
in dimension one and two, as presented in the next subsection.

\subsection{Simple odd function-germs and their odd deformations} \label{subsec:classification}

Here we deduce the normal forms and their mini-versal deformations for the simple odd singularities of
function germs in one and two variables. We have chosen a particular notation for each.
We start with the cases in one-variable. The results are obtained straightforwardly and are given in the next theorem and
corollary. The following theorem and corollary deal with the cases in two variables.

\begin{thm}
Let $g\in \mathcal E_1^{odd}$. Then $g$ is $\mathcal R^{odd}$-simple if, and
only if, $g$ is $\mathcal R^{odd}$-equivalent to one of the following
function-germs at $0$:
$$
A_{2k/2}: \ x\mapsto x^{2k+1}, \ \text{for} \ k=1,2,\cdots
$$
\end{thm}

\begin{cor}\label{A-k-def}
For $k \geq1$,  $\mathcal R^{odd}$-miniversal  deformation of  $A_{2k/2}$ is
\[ G(x,\lambda_1,\cdots,\lambda_k) \ = \
x^{2k+1}+\sum_{j=1}^{k}\lambda_{j}x^{2j-1} \ . \]
\end{cor}

\begin{thm}\label{odd-class-2}
Let $g \in \mathcal E_2^{odd}$. Then $g$ is $\mathcal R^{odd}$-simple if, and
only if, $g$ is $\mathcal R^{odd}$-equivalent to one of the following
function-germs at $0$:
$$
D_{2k/2}^{\pm}: \ (x_1,x_2)\mapsto x_1^{2}x_2\pm x_2^{2k-1}, \
\text{for} \ k=2,3,\cdots
$$
$$
E_{8/2}: \ (x_1,x_2)\mapsto x_1^{3}+ x_2^{5},
$$
$$
J_{10/2}^{\pm}: \ (x_1,x_2)\mapsto x_1^{3}\pm x_1x_2^{4}.
$$
$$
E_{12/2}: \ (x_1,x_2)\mapsto x_1^{3}+x_2^{7}.
$$
\end{thm}
\begin{proof} The procedure is the systematic usage of the complete transversal method (\cite{BKdP}, \cite{K}) at the level of jets and then usage of the finite determinacy theorem. In our context, the complete transversal  is a subspace $T$ of $\mathcal M_m^{2k+1(odd)}$ such that
\begin{equation}
\mathcal M_2^{2k+1(odd)}\subset L\mathcal R^{odd}_1\cdot g + T + \mathcal M_2^{2k+3(odd)},
\end{equation}
where $\mathcal R_1^{odd}$ is the subgroup of $\mathcal R^{odd}$ whose elements have $1$-jet equal to identity, and $L\mathcal R^{odd}_1\cdot g$ is the tangent space to the $\mathcal R_1^{odd}$-orbit of $g$ at $g$.

We start with the 3-jet of $g$, which is also the starting point of the classification without
symmetry. Since linear changes of coordinates are ${\mathbb Z}_2$-equivariant, it follows that,
at this level, the results here are precisely the same as in the context without symmetry. Therefore,
as it is well known,
a nonzero cubic polynomial in two variables is linearly equivalent to one of the following types:
\begin{eqnarray}\label{j3D4}
x_1^2x_2 \pm x_2^3 \\ \label{j3Dk}
x_1^2x_2 \\ \label{j3E8}
x_1^3
\end{eqnarray}

First, assuming that  $j^3_0 g$ is of form (\ref{j3D4}), the $\mathcal R_1^{odd}$ tangent space of the orbit of (\ref{j3D4}) is $\mathcal M_2^{5(odd)}$. The complete transversal is empty in this case and $g$ is finitely $\mathcal R^{odd}$-determined and $\mathcal R^{odd}$-equivalent to (\ref{j3D4}).

Now, assume that $j^3_0 g$ has form (\ref{j3Dk}), whose orbit has
\[ \mathcal E^{even}_2 \cdot \{x_1^5,x_1^4x_2,x_1^3x_2^2,x_1^2x_2^3, x_1x_2^4\} \]
as its $\mathcal R_1^{odd}$ tangent space. So the complete transversal is $T=\mathbb R\{x_2^5\}$.
Hence, $j^5_0 g$ is $\mathcal R_1^{odd}$- equivalent to $x_1^2x_2+ax_2^5$ and it is easy to see that  if \  $a>0$ then $j^5_0 f$ is $\mathcal R^{odd}$-equivalent to $x_1^2x_2+x_2^5$, and if  \ $a<0$ then
$j^5_0 f$ is $\mathcal R^{odd}$-equivalent to $x_1^2x_2-x_2^5$. In the next step we check that the $\mathcal R_1^{odd}$ tangent space to the orbit of both of these germs is $\mathcal M_2^{5(odd)}$. So the complete transversal is empty and $g$ is finitely $\mathcal R^{odd}$-determined and $\mathcal R^{odd}$-equivalent to
$x_1^2x_2\pm x_2^5$.
If $a=0$, then  $T = \mathbb R\{x_2^7\}$ and $j^7_0 g$ is $\mathcal R_1^{odd}$-equivalent to $x_1^2x_2+b x_2^7$. Proceeding inductively, we obtain that if $j^3g_0$ has the form (\ref{j3Dk}) and $g$ is finitely $\mathcal R^{odd}$-determined then $g$ is $\mathcal R^{odd}$-equivalent to $x_1^2x_2\pm x_2^{2k+1}$ for $k\ge 2$.

Finally, assume that $j^3_0 g$ has the form (\ref{j3E8}). In this case,  $T=\mathbb R\{x_1x_2^4,x_2^5\}$ and $j^5_0 g$ is $\mathcal R_1^{odd}$-equivalent to $x_1^3+a x_2^5+b x_1x_2^4$.

If $a\ne 0$, then $j^5_0g$ is $\mathcal R^{odd}$-equivalent to $x_1^3+x_2^5+b x_1x_2^4$ and
\[\mathcal E_2^{even} \cdot  \{x_2^5,x_1x_2^4,x_1^2x_2,x_1^3\} \]
is its tangent space. So its dimension does not depend on $b$ and it contains the germ of  $x_1 x_2^4$. It then
follows from Mather's lemma that $j^5_0 g$ is $\mathcal R^{odd}$-equivalent to $x_1^3+x_2^5$.
As  next step  we obtain that $g$ is  finitely $\mathcal R^{odd}$-determined. Then, $g$ is $\mathcal R^{odd}$-equivalent to $x_1^3+x_2^5$.

If $a=0$ and $b\ne 0$, then $j^5_0g=x_1^3\pm x_1 x_2^4$ and  $T=\mathbb R\{x_2^7\}$. Then $j^7_0 g$ is $\mathcal R_1^{odd}$-equivalent to $x_1^3\pm x_1 x_2^4+a x_2^7$. But $L\mathcal R^{odd}g$ is given by
\[ \mathcal E_2^{even} \cdot  \{x_2^7,x_1x_2^4,3x_1^2x_2\pm x_2^5,x_1^3\}.\]
Its dimension independs on $a$ and it contains $x_2^7$. By Mather's lemma, $j^7_0 g$ is $\mathcal R^{odd}$-equivalent to $x_1^3\pm x_1x_2^4$. As in the previous case, we  find  that $g$ is  finitely $\mathcal R^{odd}$-determined, so is $\mathcal R^{odd}$-equivalent to $x_1^3\pm x_1x_2^4$.

If $a=b=0$, then $j^5_0 g=x_1^3$. Thus, the complete transversal is $T=\mathbb R\{x_1x_2^6, x_2^7\}$. It means that $j^7_0 g$ is $\mathcal R^{odd}$-equivalent to $x_1^3+cx_1x_2^6+dx_2^7$. If $d\ne 0$ we may assume that $j^7_0 g=x_1^3+cx_1x_2^6+x_2^7$. But $L\mathcal R^{odd}g$ is
\[ \mathcal E_2^{even} \cdot  \{x_2^7,x_1x_2^6,x_1^2x_2,x_1^3\}. \]
 Its dimension independs on $c$ and it contains $x_1x_2^6$. By Mather's lemma,  $j^7_0 g$ is $\mathcal R^{odd}$-equivalent to $x_1^3+x_2^7$. As next step  we obtain that $g$ is  finitely $\mathcal R^{odd}$-determined, so is $\mathcal R^{odd}$-equivalent to $x_1^3+x_2^7$.

If $d=0$ and $c\ne 0$, we may assume $j^7_0g=x_1^3\pm x_1x_2^6$. The complete transversal is $T=\mathcal R\{x_2^9\}$. So $j^9_0g$ is $\mathcal R^{odd}_1$-equivalent to $x_1^3\pm x_1x_2^6+ax_2^9$.
But $x_2^9\notin L\mathcal R^{odd}j^9_0 g$. By Mather's lemma, $c$  is a modulus.
\end{proof}

From  Theorem \ref{inf-versal} and  the
proof of Theorem~\ref{odd-class-2}, we obtain:

\begin{cor}\label{D_k-versal}\label{E_8-versal}\label{J_10-versal}\label{E_12-versal}
The $\mathcal R^{odd}$-miniversal deformation of the odd-simple map-germs are given by:
$$D_{2k/2}^{\pm} \ : \ F(x_1,x_2,\lambda_1,\cdots,\lambda_k)\equiv \ x_1^{2}x_2\pm
x_2^{2k-1}+\lambda_1x_1+\sum_{i=2}^{k}\lambda_{i}x_2^{2i-3}.
$$
$$E_{8/2} \ : \ F(x_1,x_2,\lambda_1,\cdots,\lambda_4)\equiv \ x_1^{3}+
x_2^{5}+\lambda_1x_1+\lambda_2x_2+\lambda_3x_1 x_2^2+\lambda_4
x_2^3.
$$
$$J_{10/2}^{\pm} \  : \ F(x_1,x_2,\lambda_1,\cdots,\lambda_5)\equiv $$
$$x_1^{3}\pm
x_1x_2^{4}+\lambda_1x_1+\lambda_2x_2+\lambda_3x_1^2 x_2+\lambda_4
x_2^2x_1+\lambda_5 x_2^3.
$$
$$E_{12/2} \  : \ F(x_1,x_2,\lambda_1,\cdots,\lambda_6)\equiv $$
$$ x_1^{3}+
x_2^{7}+\lambda_1x_1+\lambda_2x_2+\lambda_3x_1 x_2^2+\lambda_4
x_2^3+\lambda_5x_1x_2^5+\lambda_6x_2^6.
$$
\end{cor}

\begin{rem} \label{remark corank 2}
The notations for the odd-simple singularities presented above have been chosen by their resemblance with the classical notations \cite{AGV} for normal forms of $\mathcal R$-singularities. In fact, $A_{2k/2}$ has the same representative as $A_{2k}$, but while the latter has codimension $2k$, the former has odd codimension $k=2k/2$.  Similarly, for $D_{2k/2}$ and $E_{8/2}$, with odd codimensions $k$ and $4$, respectively, for which the corresponding $\mathcal R$-singularities $D_{2k}$ and $E_8$ have codimensions $2k$ and $8$, respectively.
The situation differs for the other odd-simple singularities. The germ of the odd codimension $6$ singularity $E_{12/2}$ is $\mathcal R$-equivalent to the codimension 12 singularity $E_{12}$, but
we stress that the latter is unimodal. Similarly for the odd codimension $5$ singularity $J^{\pm}_{10/2}$ in
comparison with codimension $10$ unimodal $\mathcal R$-singularity $J_{10}$.
\end{rem}

\section{Simple stable singularities of Wigner caustic on shell}

From classical results (\cite{AGV}) we know that Lagrangian equivalence of Lagrangian maps corresponds to stable fibred $\mathcal R^+$-equivalence of their generating families (see Remark \ref{Lag-stable}). Thus we introduce the following definition in the $\mathbb Z_2$-symmetric case.
\begin{defn}
Let $L$ and $\tilde L$ be germs at $(0,0)\in \mathbb R^{2m}$ of
Lagrangian submanifolds of the affine symplectic space. The germs at $(0,0)$ of
Wigner caustics on shell ${\bf E}_{{1}/{2}}(L)$ and ${\bf E}_{{1}/{2}}(\tilde L)$
are {\bf Lagrangian equivalent} if germs at $(0,0,0)\in  \mathbb
R^m\times \mathbb R^{2m}$ of the corresponding odd generating families
$F$ and $\tilde F$ are fibred $\mathcal R^{odd}$-equivalent.
\end{defn}

From Remark \ref{hidden}, this means equivalence of $\mathbb Z_2$-symmetric germs of Wigner caustics. The following definition specializes to this $\mathbb Z_2$-symmetric context  the well-known  fact (\cite{AGV}) that stability of Lagrangian maps corresponds to versality of generating families (Remark \ref{Lag-stable}).

\begin{defn} A germ of Wigner caustic on shell is {\bf stable} if its generating family is an $\mathcal R^{odd}$-versal deformation of an odd  function-germ, and it is  {\bf simple stable} if its generating family is an $\mathcal R^{odd}$-versal deformation of an $\mathcal R^{odd}$-odd simple function-germ.
\end{defn}

Notice that any odd function-germ $f\in \mathcal M_m^{3(odd)}$ can be written as $f(\beta)\equiv \frac{1}{2}\left(S(\beta)-S(-\beta)\right)$ for some  $S\in \mathcal M_m^3$, implying the following:
\begin{prop}
For any $f\in \mathcal M_m^{3(odd)}$ there exists $S\in \mathcal M_m^3$ such that the generating family $F$ of the form (\ref{gf}) is an odd deformation of $f$.
\end{prop}
By Theorem \ref{inf-versal} we obtain the following corollary.
\begin{cor}\label{versalF}
The germ of a generating family $F$ of the form (\ref{gf}) is an $\mathcal R^{odd}$-versal deformation if and only if
\begin{eqnarray}
\quad\quad \mathcal M_m^{3(odd)}&=&\mathcal E_m^{even}\left\{\beta_i\left(\frac{\partial S}{\partial q_j}(\beta)+\frac{\partial S}{\partial q_j}(-\beta)\right): \ i,j=1,\cdots m\right\}+ \label{versalF1} \\
&&\mathbb R\left\{\frac{\partial S}{\partial q_j}(\beta)-\frac{\partial S}{\partial q_j}(-\beta): \ j=1,\cdots m\right\}. \nonumber
\end{eqnarray}
\end{cor}
From Corollary \ref{versalF} we get the following realization theorem.
\begin{thm}\label{realization}
Let $f\in \mathcal M_m^{3(odd)}$ be a finitely determined germ.
Then there exists $S\in \mathcal M_m^3$ such that the generating family $F$ of the form (\ref{gf}) is an $\mathcal R^{odd}$-versal deformation of $f$ if and only if there exist $h_1$,$\cdots$,$h_m$ in $\mathcal M_m^{3(odd)}$ such that
\begin{equation}\label{verf}
\mathcal M_m^{3(odd)}=L \mathcal R^{odd} f + \mathbb R \{h_1,\cdots,h_m\}
\end{equation}
and $\sum_{i=1}^m h_i(\beta_1,\cdots,\beta_m)d\beta_i$ is a germ of closed $1$-form.
\end{thm}
\begin{proof} First, notice that any fuction-germ $S\in \mathcal E_m$ can be decomposed into $S=S^++S^-$, where $S^+\in \mathcal E_m^{even}$, $S^-\in \mathcal E_m^{odd}$ are given in the following way $S^+(\beta)\equiv \frac{1}{2}(S(\beta)+S(-\beta))$, $S^-(\beta)\equiv \frac{1}{2}(S(\beta)-S(-\beta))$.
Then the versality condition (\ref{versalF1}) of $F$ given by (\ref{gf}) has the  form
$$
\mathcal M_m^{3(odd)}=L\mathcal R^{odd} S^-+
\mathbb R\left\{\frac{\partial S^+}{\partial \beta_j}(\beta): \ j=1,\cdots m\right\}.
$$
From the above, $f$ must be equal to $S^-$ and the germ of a $1$-form $\sum_{j=1}^m\frac{\partial S^+}{\partial \beta_j}(\beta)d\beta_j$ is closed since it is just $dS^+$. On the other hand if condition (\ref{verf}) is satisfied and $\alpha=\sum_{i=1}^m h_i(\beta_1,\cdots,\beta_m)d\beta_i$ is a germ of closed $1$-form then it is obvious that there exists such a function-germ $g \in \mathcal E_m^{even}$ such that $\alpha=dg$. So we take $S=f+g$.
\end{proof}

It follows from Theorem \ref{simpleless3} that simple singularities for the Wigner caustic on shell of a Lagrangian submanifold can be realized only for curves in $\mathbb R^2$ and surfaces in $\mathbb R^4$.
Thus, first we apply Theorem \ref{realization} to check which versal deformations of simple odd singularities are realizable as a generating family of the form (\ref{gf}).
\begin{cor}
$\mathcal R^{odd}$-versal  deformations of $A_{2/2}$, $A_{4/2}$ (for $m=1$) and $D_{4/2}^{\pm}$, $D_{6/2}^{\pm}$ $D_{8/2}^{\pm}$, $E_{8/2}$ (for $m=2$) are realizable as generating families of form (\ref{gf}).

 $\mathcal R^{odd}$-versal  deformations of $A_{2k/2}$ for $k>2$ (and for $m=1$) and $D_{2k/2}$ for $k>4$, $J_{10/2}^{\pm}$ and $E_{12/2}$ (for $m=2$) are not realizable as generating families of form (\ref{gf}).
\end{cor}

\begin{proof}
First notice that if the codimension of the singularity is greater than $2m$ then the $\mathcal R^{odd}$-versal deformation of it is not realizable by a generating family of the form (\ref{gf}). This proves the second statement.
Since any smooth $1$-form on $\mathbb R$ is closed this is the only restriction for $m=1$.
The realization of $D_{4/2}^{\pm}$ is obvious. For the others singularities we apply Theorem \ref{realization} in the following way: for $D^{\pm}_{6/2}$ take $h_1(\beta)\equiv 0$ and $h_2(\beta)\equiv \beta_2^3$, for $D^{\pm}_{8/2}$ take $h_1(\beta)\equiv \beta_2^3$ and $h_2(\beta)\equiv \beta_2^5+3\beta_1\beta_2^2$,
and for $E_{8/2}$ take $h_1(\beta)\equiv \beta_2^3$ and $h_2(\beta)\equiv 3\beta_1\beta_2^2.$
\end{proof}

\subsection{The Wigner caustic on shell of a Lagrangian curve.}

Let $L$ be the germ at $(0,0)$ of a  curve on
symplectic affine plane $(\mathbb R^2, \omega=dp\wedge dq)$ and, without loss of generality,  assume that
$L$ is generated by a  function-germ $S\in \mathcal M_1^3 \subset \mathcal E_1$ in the usual way given by  (\ref{usual}), $i=1$.
\begin{thm}\label{A2-A4} Let $F$ of form (\ref{gf}) be the generating family of
$\mathcal L$.

If $\frac{d^3 S}{d q^3}(0)\ne 0$, $F$ is fibred $\mathcal R^{odd}$-equivalent to the $\mathcal
R^{odd}$-versal deformation of $A_{2/2}$ : $(\beta,p,q)\mapsto
\beta^3+p\beta$.

If $\frac{d^3 S}{d q^3}(0)=0$, $\frac{d^4 S}{d q^4}(0)\ne 0$ and
$\frac{d^5 S}{d q^5}(0)\ne 0$, $F$ is $\mathcal R^{odd}$-equivalent to the $\mathcal
R^{odd}$-versal deformation of $A_{4/2}$ :
$(\beta,p,q)\mapsto \beta^5+q\beta^3+p\beta$.
\end{thm}

\begin{proof}
From (\ref{gf}),
$\frac{\partial^k F}{\partial
\beta^k}(\beta,0,0)=\frac{1}{2}\left(\frac{d^k S}{d
q^k}(\beta)+(-1)^{k+1}\frac{d^k S}{d q^k}(-\beta)\right).$
Thus $\frac{d^3 S}{d q^3}(0)\ne 0$ implies that $\frac{d^3 F}{d
\beta^3}(0,0,0)\ne 0$ and $F|_{\mathbb R\times \{0\}\times \{0\}}
\in \mathcal M_1^3$, since $S\in \mathcal M_1^3$.
Therefore, $F$ is an odd deformation of $A_{2/2}$. By Theorems
\ref{inf-versal} and \ref{A-k-def} we obtain that $F$ is $\mathcal
R^{odd}$-equivalent to $\mathcal R^{odd}$-versal deformation of $A_{2/2}$:
$(\beta,p,q)\mapsto \beta^3+p\beta$, since $\frac{\partial
F}{\partial p}(\beta,0,0)=-\beta$.

If $S\in \mathcal M_1^4$ and $\frac{d^5 S}{d q^5}(0)\ne 0$ then
 $\frac{\partial^k F}{\partial \beta^k}(0,0,0)=0$
for $k<5$ and $\frac{\partial^5 F}{\partial \beta^5}(0,0,0)\ne 0$
and consequently $F$ is an odd deformation of $A_{4/2}$.
By direct calculation,
$\frac{\partial^{k+1} F}{\partial \beta^k
\partial q}(\beta,0,0)=\frac{1}{2}\left(\frac{d^{k+1} S}{d
q^{k+1}}(\beta)+(-1)^{k+1}\frac{d^{k+1} S}{d
q^{k+1}}(-\beta)\right).$

Then $\frac{\partial^{k+1} F}{\partial \beta^k
\partial q}(0,0,0)=0$ for $k<3$ and $\frac{\partial^{4} F}{\partial
\beta^{3} \partial q}(0,0,0)=\frac{d^4 S}{d q^4}(0)$. But
$\frac{\partial F}{\partial p}(\beta,0,0)=-\beta$. So if
$\frac{d^4 S}{d q^4}(0)\ne 0$  we obtain by  Theorem
\ref{inf-versal} and Corollary \ref{A-k-def} that $F$ is $\mathcal
R^{odd}$-equivalent to $\mathcal R^{odd}$-miniversal deformation of $A_{4/2}$:
$(\beta,p,q)\mapsto \beta^5+q\beta^3+p\beta$.
\end{proof}

\begin{cor}\label{GI} {\emph ({\bf Geometric interpretation})}
If the curvature of the germ of a Lagrangian curve $L$ does not
vanish at $(p_0,q_0)\in L$, then the germ at $(p_0,q_0)$ of the Wigner caustic
on shell consists of $L$ only and is Lagrangian stable. All germs of
Wigner caustics of Lagrangian curves at such points are Lagrangian
equivalent.

If, at  $(p_0,q_0)\in L$, the curvature of the germ of a Lagrangian curve
$L$ vanishes but the first and the second derivatives of the curvature
do not vanish, then the germ at $(p_0,q_0)$ of the Wigner caustic on shell consists
of two components: $L$ and the germ at $(p_0,q_0)$ of a
$1$-dimensional smooth submanifold with boundary $(p_0,q_0)$, which is
$1$-tangent to $L$ at $(p_0,q_0)$
and is simple
stable. Any germ of the Wigner caustic in such a point is Lagrangian
equivalent to the following germ at $0$:
$$
\left\{(p,q)\in \mathbb R^2: p=0 \right\}\cup \left\{(p,q)\in
\mathbb R^2: p=-\frac{27}{50}q^2, q\le 0 \right\}.
$$
The germs of the Wigner caustics of $L$ at points of $L$ which do
not satisfy the above conditions are not stable.
\end{cor}
\begin{proof} This is an obvious corollary of Theorem \ref{A2-A4}, because the curvature of a curve $L$ described by (\ref{usual}), $i=1$, is given by
$
\kappa\left(\frac{dS}{dq}(q),q\right)={\frac{d^3S}{dq^3}(q)}\Big{/}{\left(1+\left(\frac{d^2S}{dq^2}(q)\right)^2\right)^{3/2}}
$
Thus  $\kappa(p_0,q_0)=\frac{d^3S}{dq^3}(q_0)$ since $\frac{d^2S}{dq^2}(q_0)=0$. If $\kappa(p_0,q_0)=0$ then
$\frac{d\kappa}{dq}(p_0,q_0)=\frac{d^4S}{dq^4}(q_0)$ and $\frac{d^2\kappa}{dq^2}(p_0,q_0)=\frac{d^5S}{dq^5}(q_0)$
\end{proof}
\begin{rem}
Although the curvature of a plane curve is not an affine invariant, the vanishing or not vanishing of the curvature is an affine invariant. Also, where the curvature is zero, the vanishing or not vanishing of its first two derivatives is also an affine invariant. Thus, Corollary \ref{GI} provides coordinate-free affine-symplectic invariant conditions for the realization of the singularities of the Wigner caustic on shell of a Lagrangian curve on the affine symplectic plane. Similar results for curves on a affine plane without a symplectic structure can be found in \cite{GWZ}, where bifurcations of affine equidistants  were studied.
\end{rem}

\subsection{The Wigner caustic on shell of a Lagrangian surface.}

Let $L$ be the germ at $0$ of a Lagrangian surface in
symplectic affine space $(\mathbb R^4, \omega=dp_1\wedge dq_1+dp_2\wedge dq_2)$ and, without loss of generality, assume that
$L$ is generated by a  function-germ $S\in \mathcal M_2^3 \subset \mathcal E_2$ by (\ref{usual}), $i=2$,
and that $F$ of form (\ref{gf}) is the generating family of $\mathcal L$.

\begin{notation}\label{partialS} To simplify the equations, we use the following:
$$
S_{i,j}=\frac{\partial^{i+j} S}{\partial q_1^i \partial q_2^j}(0,0) \ , \ S_{i,j}(q)=\frac{\partial^{i+j} S}{\partial q_1^i \partial q_2^j}(q_1,q_2).
$$
Then, the $3$-jet of $S$ at $0$ has the form
$$j^3_0S=\frac{1}{6}
S_{3,0}q_1^3+\frac{1}{2} S_{2,1}q_1^2q_2+\frac{1}{2}S_{1,2}q_1q_2^2+\frac{1}{6}S_{0,3}q_2^3$$
and the {\bf discriminant} of $j^3_0S$ has the following form $\Delta(j^3_0S)=$
$$
\frac{1}{48} \left(3S_{1,2}^2 S_{2,1}^2-4 S_{0,3}S_{2,1}^3
-4 S_{1,2}^3 S_{3,0} -S_{0,3}^2S_{3,0}^2+6 S_{0,3}S_{1,2}S_{2,1}S_{3,0}\right)
$$
\end{notation}

\begin{thm}\label{thmD4-}
If $\Delta(j^3_0S)>0$, $F$ is $\mathcal R^{odd}$-equivalent to the $\mathcal
R^{odd}$-versal deformation of  $D_{4/2}^-$ : $(\beta_1,\beta_2,p,q)\mapsto
\beta_1^2\beta_2-\beta_2^3+p_1\beta_1+p_2\beta_2$.

If $\Delta(j^3_0S)<0$, $F$ is $\mathcal R^{odd}$-equivalent to the $\mathcal
R^{odd}$-versal deformation of $D_{4/2}^+$ : $(\beta_1,\beta_2,p,q)\mapsto
\beta_1^2\beta_2+\beta_2^3+p_1\beta_1+p_2\beta_2$.
\end{thm}
\begin{proof}
 By (\ref{f}) we get that $j^3_0f=j^3_0S$. If $\Delta(j^3_0S)>0$, by linear change of coordinates we can reduce $j^3_0f$ to  $\beta_1^2\beta_2-\beta_2^3$. Then repeating the arguments in the proof of Theorem \ref{odd-class-2} it is easy to see that  $f$ is $\mathcal R^{odd}$-equivalent to $D_{4/2}^-$ singularity. By Theorem \ref{versalF} it is easy to see that (\ref{gf}) is an $\mathcal R^{odd}$-versal deformation of $f$.
By Corollary \ref{D_k-versal} we get the result. The case $\Delta(j^3_0S)<0$ is analogous.
\end{proof}
\begin{lem}\label{lem-non-pos}
$S_{3,0}S_{1,2}-S_{2,1}^2\le 0$ and $S_{0,3}S_{2,1}-S_{1,2}^2\le 0$, if $\Delta(j^3_0S)=0$.
\end{lem}
\begin{proof}
The condition $\Delta(j^3_0S)=0$ implies that $w(t)=\frac{1}{6}
S_{3,0}t^3+\frac{1}{2} S_{2,1}t^2+\frac{1}{2}S_{1,2}t+\frac{1}{6}S_{0,3}$ and $v(t)=\frac{1}{6}
S_{3,0}+\frac{1}{2} S_{2,1}t+\frac{1}{2}S_{1,2}t^2+\frac{1}{6}S_{0,3}t^3$ have real roots of multiplicity greater than $1$. Thus polynomials $\frac{dw}{dt}(t)=\frac{1}{2}
S_{3,0}t^2+S_{2,1}t+\frac{1}{2}S_{1,2}$ and $\frac{dv}{dt}(t)=\frac{1}{2}
S_{2,1}+S_{1,2}t+\frac{1}{2}S_{0,3}t^2$ have real roots. So their discriminants are nonnegative.
\end{proof}
\begin{notation} Now we introduce the following abbreviations:
$$r_1=\frac{S_{2,1}S_{1,2}-S_{3,0}S_{0,3}}{2(S_{3,0}S_{1,2}-S_{2,1}^2)}, \ \
r_2=\frac{S_{3,0}^2S_{0,3}-S_{3,0}S_{2,1}S_{1,2}+3S_{2,1}^3}{S_{3,0}S_{1,2}-S_{2,1}^2}$$
$$\sigma_{0,n}=\frac{\sum_{k=0}^n \left(^n_k\right)S_{k,n-k}r_1^k}{(S_{3,0}r_1-r_2)^n} \ \ \text{for} \ \ n=5,7$$
$$\tilde{r}_1=\frac{S_{2,1}S_{1,2}-S_{3,0}S_{0,3}}{2(S_{0,3}S_{2,1}-S_{1,2}^2)}, \ \
\tilde{r}_2=\frac{S_{0,3}^2S_{3,0}-S_{0,3}S_{1,2}S_{2,1}+3S_{1,2}^3}{S_{0,3}S_{2,1}-S_{1,2}^2}$$
$$\sigma_{n,0}=\frac{\sum_{k=0}^n \left(^n_k\right)S_{n-k,k}\tilde{r}_1^k}{(S_{0,3}\tilde{r}_1-\tilde{r}_2)^n} \ \ \text{for} \ \ n=5,7$$
\end{notation}
\begin{thm}\label{thmD6+}
Assume  $S$ satisfies condition (\ref{versalF1}) and $\Delta(j^3_0S)=0$.
Consider the following pair of conditions:
\begin{equation}\label{dp}
S_{3,0}S_{1,2}-S_{2,1}^2<0,
\end{equation}
\begin{equation}\label{dp'}
S_{0,3}S_{2,1}-S_{1,2}^2<0.
\end{equation}

If (\ref{dp}) is satisfied and $\sigma_{0,5}>0$, or
(\ref{dp'}) is  satisfied and $\sigma_{5,0}>0$, then
 $F$ is $\mathcal R^{odd}$-equivalent to the $\mathcal
R^{odd}$-versal deformation of  $D_{6/2}^+$ : $(\beta_1,\beta_2,p,q)\mapsto
\beta_1^2\beta_2+\beta_2^5+p_1\beta_1+p_2\beta_2+q_1\beta_2^3$.

If (\ref{dp}) is satisfied and $\sigma_{0,5}<0$, or
(\ref{dp'}) is  satisfied and $\sigma_{5,0}<0$, then
$F$ is $\mathcal R^{odd}$-equivalent to the $\mathcal
R^{odd}$-versal deformation of  $D_{6/2}^-$ : $(\beta_1,\beta_2,p,q)\mapsto
\beta_1^2\beta_2-\beta_2^5+p_1\beta_1+p_2\beta_2+q_1\beta_2^3$.
\end{thm}
\begin{proof}
First we assume $\Delta(j^3_0S)=0$ and condition (\ref{dp}), with $\sigma_{0,5}>0$.  Then
we get $j^3_0f=j^3_0S=(\beta_1-r_1\beta_2)^2(S_{3,0}\beta_1-r_2\beta_2)=\tilde{\beta}_1^2\tilde{\beta}_2$, where $(\tilde{\beta}_1,\tilde{\beta}_2)=(\beta_1-r_1\beta_2,S_{3,0}\beta_1-r_2\beta_2)$ forms the coordinate system on $\mathbb R^2$, since by condition (\ref{dp}) $r_1\ne r_2/S_{3,0}$. $\sigma_{0,5}>0$ is equivalent to $\frac{\partial^5 f}{\partial \tilde{\beta}_2^5}(0)>0$. Thus, $f$ is $\mathcal R^{odd}$-equivalent to $D_{6/2}^+$. By Theorem \ref{versalF} we obtain  that $F$ is an $\mathcal R^{odd}$-versal deformation of $f$ since $S$ satisfies (\ref{versalF1}).
If $\Delta(j^3_0S)=0$ and (\ref{dp'}) is satisfied with $\sigma_{5,0}>0$, then we repeat in the same way using the coordinate system $(\tilde{\beta}_1,\tilde{\beta}_2)=(\beta_2-\tilde{r}_1\beta_1,S_{0,3}\beta_2-\tilde{r}_2\beta_1)$. The cases  (\ref{dp}) and $\sigma_{0,5}<0$, or (\ref{dp'}) and $\sigma_{5,0}<0$, are analogous.
\end{proof}
\begin{thm}\label{thmD8+}
Assume  $S$ satisfies condition (\ref{versalF1}) and $\Delta(j^3_0S)=0$.

If (\ref{dp}) holds, $\sigma_{0,5}=0$ and $\sigma_{0,7}>0$, or, if
(\ref{dp'}) holds, $\sigma_{5,0}=0$ and $\sigma_{7,0}>0$, then
$F$ is $\mathcal R^{odd}$-equivalent to the $\mathcal
R^{odd}$-versal deformation of  $D_{8/2}^+$ : $(\beta_1,\beta_2,p,q)\mapsto
\beta_1^2\beta_2+\beta_2^7+p_1\beta_1+p_2\beta_2+q_1\beta_2^3+q_2\beta_2^5$.

If (\ref{dp}) holds, $\sigma_{0,5}=0$ and $\sigma_{0,7}<0$, or, if (\ref{dp'}) holds, $\sigma_{5,0}=0$ and $\sigma_{7,0}<0$, then
$F$ is $\mathcal R^{odd}$-equivalent to the $\mathcal
R^{odd}$-versal deformation of $D_{8/2}^-$ : $(\beta_1,\beta_2,p,q)\mapsto
\beta_1^2\beta_2-\beta_2^7+p_1\beta_1+p_2\beta_2+q_1\beta_2^3+q_2\beta_2^5$.
\end{thm}
\begin{proof}
First assume $\Delta(j^3_0S)=0$,   condition (\ref{dp}) is satisfied and $\sigma_{0,5}=0$, $\sigma_{0,7}>0$. As in the proof of Theorem \ref{thmD6+}
we get $j^3_0f=j^3_0S=\tilde{\beta}_1^2\tilde{\beta}_2$, where $(\tilde{\beta}_1,\tilde{\beta}_2)=(\beta_1-r_1\beta_2,S_{3,0}\beta_1-r_2\beta_2)$ and $\frac{\partial^5 f}{\partial \tilde{\beta}_2^5}(0)=0$ and $\frac{\partial^7 f}{\partial \tilde{\beta}_2^7}(0)>0$, since $\sigma_{0,5}=0$ and $\sigma_{0,7}>0$. Thus $f$ is $\mathcal R^{odd}$-equivalent to $D_{8/2}^+$.  By Theorem \ref{versalF}, $F$ is an $\mathcal R^{odd}$-versal deformation of $f$ since $S$ satisfies (\ref{versalF1}).
If $\Delta(j^3_0S)$ vanishes,   condition (\ref{dp'}) is satisfied and $\sigma_{5,0}=0$, $\sigma_{7,0}>0$, we repeat  using the coordinate system $(\tilde{\beta}_1,\tilde{\beta}_2)=(\beta_2-\tilde{r}_1\beta_1,S_{0,3}\beta_2-\tilde{r}_2\beta_1)$. The case
(\ref{dp}), $\sigma_{0,5}=0$ and $\sigma_{0,7}<0$, and the case  (\ref{dp'}), $\sigma_{5,0}=0$ and $\sigma_{7,0}<0$, are worked out analogously.
\end{proof}
\begin{thm}\label{thmE8}
Assume $S$ satisfies condition (\ref{versalF1}) and $\Delta(j^3_0S)=0$.
If either of the following two  conditions are satisfied,
\begin{equation}\label{dp0}
S_{3,0}S_{1,2}-S_{2,1}^2=0, \  S_{3,0}\ne 0, \  \sum_{k=0}^5\left(^5_k\right)S_{k,5-k}\left(-S_{2,1}\right)^k\left(S_{3,0}\right)^{5-k}\ne0,
\end{equation}
\begin{equation}\label{dp0'}
S_{0,3}S_{2,1}-S_{1,2}^2=0, \ S_{0,3}\ne 0, \  \sum_{k=0}^5\left(^5_k\right)S_{5-k,k}\left(-S_{1,2}\right)^k\left(S_{0,3}\right)^{5-k}\ne0,
\end{equation}
then $F$ is $\mathcal R^{odd}$-equivalent to the $\mathcal
R^{odd}$-versal deformation of $E_{8/2}$ : $(\beta_1,\beta_2,p,q)\mapsto
\beta_1^3+\beta_2^5+p_1\beta_1+p_2\beta_2+q_1\beta_1\beta_2^2+q_2\beta_2^3$.
\end{thm}
\begin{proof}
First we assume that $\Delta(j^3_0S)=0$  and  condition (\ref{dp0}) is satisfied. It implies that
we get $j^3_0f=j^3_0S=\tilde{\beta}_1^3$, where $(\tilde{\beta}_1,\tilde{\beta}_2)=\left(\left(\frac{S_{3,0}}{6}\right)^{1/3}\left(\beta_1-\frac{S_{2,1}}{S_{3,0}}\beta_2\right),\beta_2\right)$ and $\frac{\partial^5 f}{\partial \tilde{\beta}_2^5}(0)\ne 0$. Thus $f$ is $\mathcal R^{odd}$-equivalent to $E_{8/2}$.  By Theorem \ref{versalF} we obtain  that $F$ is an $\mathcal R^{odd}$-versal deformation of $f$ since $S$ satisfies (\ref{versalF1}).
If $\Delta(j^3_0S)=0$ and   condition (\ref{dp0'}) is satisfied, we repeat with $(\tilde{\beta}_1,\tilde{\beta}_2)=\left(\left(\frac{S_{0,3}}{6}\right)^{1/3}\left(\beta_2-\frac{S_{1,2}}{S_{0,3}}\beta_1\right),\beta_1\right)$.
\end{proof}
\begin{rem}
These are all odd-simple singularities that can be realized as singularities of on-shell Wigner caustics of Lagrangian submanifolds in affine-symplectic space. The odd-simple singularities $J_{10/2}^{\pm}$ and $E_{12/2}$ cannot be realized in this way because their codimensions are too big for a Lagrangian surface in affine-symplectic $4$-space. On the other hand, for higher dimensional Lagrangian submanifolds in affine-symplectic space, the necessary number of variables for the generating families of on-shell Wigner caustics is at least $3$ (see Remark \ref{noquadratic}).
\end{rem}

\subsection{Geometric interpretation} Finally, we provide the geometric interpretation of each realization condition for simple stable Lagrangian singularities of on-shell Wigner caustics of Lagrangian surfaces. This also provides affine-invariant descriptions for such realization conditions, which were presented in a particular coordinate system, in Theorems \ref{thmD4-}-\ref{thmE8}. Similar results for surfaces on a affine $4$-space without a symplectic structure can be found in \cite{GJ}, where geometry of surfaces through the contact map was studied.
Background for extrinsic geometry of surfaces in euclidean $4$-space can be found in \cite{MRR}.  Here, we merely adapt it to the case of Lagrangian surfaces in affine-symplectic $4$-space. Recall Notation \ref{partialS}.

Then, for the canonical euclidean metric in $\mathbb R^4$, the matrix of the second fundamental form at $(p,q)$  of $L$ can be written as follows:
$$II_{(p,q)}=
\left[
\begin{array}{ccc}
S_{3,0}(q)&S_{2,1}(q)&S_{1,2}(q)\\
S_{2,1}(q)&S_{1,2}(q)&S_{0,3}(q)\\
\end{array}
\right]
$$
from which is defined the following determinant:
$$
\Delta_L(p,q)=
\frac{1}{4}\det\left[
\begin{array}{cccc}
S_{3,0}(q)&2S_{2,1}(q)&S_{1,2}(q)&0\\
0&S_{3,0}(q)&2S_{2,1}(q)&S_{1,2}(q)\\
S_{2,1}(q)&2S_{1,2}(q)&S_{0,3}(q)&0\\
0&S_{2,1}(q)&2S_{1,2}(q)&S_{0,3}(q)\\
\end{array}
\right]
$$
and it is easy to see that
\begin{equation}\label{Deltadiscrim}
\Delta_L(p,q)=-16\Delta(j^3_qS).
\end{equation}
Also, the Gaussian curvature at $(p,q)\in L$ is given by the  formula
\begin{equation}\label{gauss}
\kappa(p,q)=S_{3,0}(q)S_{1,2}(q)-(S_{2,1}(q))^2+S_{2,1}(q)S_{0,3}(q)-(S_{1,2}(q))^2.
\end{equation}

In extrinsic geometry of surfaces in euclidean $\mathbb R^4$, $\Delta_L$ and $\kappa$ are both invariant under the action of the euclidean group of isometries on  $\mathbb R^4$, but neither is invariant under the action of the whole affine group on $\mathbb R^4$.  The same is true if we restrict to the action of the affine-symplectic group on symplectic $\mathbb R^4$. However, although neither  $\Delta_L$ nor $\kappa$ are affine-symplectic invariants, the following propositions allow us to use them for classifying points in a Lagrangian surface of  symplectic $\mathbb R^4$.

\begin{prop}\label{signdelta} The sign ($ \ >0 \ ,\ <0 \ , \ =0 \ $) of $\Delta_L$  is an affine (and therefore affine-symplectic) invariant. \end{prop}
\begin{proof} The proof follows from the following two statements:\\ (i) The sign of $\Delta_L$ stratifies the singularities of height functions $h[\iota]:L\times S^3\to\mathbb R \ , \ (m,v)\mapsto\langle \iota(m),v\rangle$, where $\iota: L\to\mathbb R^4$ is an embedding, $S^3\subset\mathbb R^4$ is the unit sphere, and $\langle\cdot , \cdot\rangle$ is the euclidean inner product in $\mathbb R^4$ (see \cite{MRR}, Lemma 3.2, which relates the sign of $\Delta_L(p,q)$ to the number of unit vectors $v$ normal to $L$ at $(p,q)\in L$ for which $(p,q)$ is a degenerate critical point of the height function $h[\iota,v]:L\to\mathbb R$).\\ (ii) The stratification of the singularities of height functions $h[\iota]$ is invariant under affine transformations (see \cite{BGT}, Proposition A.4, which relates singularities of height functions to contact with hyperplanes and shows that the stratification of these contacts is affine invariant).
\end{proof}

Recall that a point $(p,q)\in L$ is called: \
(i)  {\bf parabolic} if $\Delta_L(p,q)=0$, \\
(ii) {\bf elliptic} if $\Delta_L(p,q)>0$, \
(iii) {\bf hyperbolic} if $\Delta_L(p,q)<0$, \ see \cite{MRR}.

\

From Proposition \ref{signdelta}, such a classification of points on $L\subset\mathbb R^4$ (with symplectic structure) is affine (and therefore affine-symplectic) invariant and, from equation (\ref{Deltadiscrim}),  we obtain the following immediate corollary of Theorem \ref{thmD4-}, which gives a geometrical characterization of singularities $D_{4/2}^{\pm}$ of the Wigner caustic on shell.
\begin{cor}\label{D4real}
Let $L$ be a Lagrangian surface. Iff $(p,q)\in L$ is a hyperbolic point, the germ of  Wigner caustic on shell at $(p,q)$ is generated by function-germ of type $D_{4/2}^-$  and it consists of $L$ only, being simple stable.
 Iff $(p,q)\in L$ is an elliptic point,  the germ of  Wigner caustic on shell  is generated by function-germ  of type $D_{4/2}^+$  and is
Lagrangian equivalent to the following simple stable germ at $0$:
\begin{equation}
{\bf E}_{{1}/{2}}(L)=\left\{(p,q)\in \mathbb R^4: 3 p_1^2=p_2^2, p_2\le 0\right\}. \nonumber
\end{equation}
\end{cor}

\

We also recall \cite{MRR} that a {\it parabolic} point $m\in M^2\subset\mathbb R^4$ is called

\noindent (i-i) an {\it inflection point of imaginary type}, if $\kappa(m)>0$,

\noindent (i-ii) an {\it inflection point of real type}, if $\kappa(m)<0$, $rank\{II_{(m)}\}=1$,

\noindent (i-iii) a {\it point of nondegenerate ellipse},  if $\kappa(m)<0$, $rank\{II_{(m)}\}=2$,

\noindent (i-iv) an {\it inflection point of flat type}, if $\kappa(m)=0$.

\

Again, we refer to \cite{MRR} where the above classification of parabolic points on $M^2\subset\mathbb R^4$ is related to the classification of singularities of height functions, which, from Proposition $A.4$ in \cite{BGT} implies:
\begin{prop}\label{parabolic}
When $\Delta(p,q)=0$, the classification of the parabolic point $m=(p,q)\in M^2\subset\mathbb R^4$ (with symplectic structure)  given by (i-i)-(i-iv) above is affine (and therefore affine-symplectic) invariant.
\end{prop}

And thus, finally, we obtain the other geometric characterizations.

\begin{cor}\label{last} If $L$ is a Lagrangian surface, then $L$ has no inflection points of real or imaginary types.
The germ of the Wigner caustic on shell at $(p,q)\in L$ has simple stable $\mathbb Z_2$-symmetric singularity generated by function-germ of  type $D_{6/2}^{\pm}$ or $D_{8/2}^{\pm}$ only if $(p,q)$ is a parabolic point of nondegenerate ellipse, and by  function-germ of type $E_{8/2}$  only if $(p,q)$ is an inflection point of flat type.
\end{cor}

\begin{proof}
(\ref{gauss}) and Lemma \ref{lem-non-pos}  imply that if $(p,q)$ is a parabolic point of a Lagrangian surface $L$ then the Gaussian curvature $\kappa(p,q)$ is nonpositive. Thus, $L$ has no inflection points of imaginary type. The simple observation, that, in the Lagrangian case (only), if $rank\{II_{(p,q)}\}=1$ then $\kappa(p,q)=0$, implies $L$ has no inflection points of real type. The second statement follows from Theorems \ref{thmD6+}-\ref{thmD8+}-\ref{thmE8} and the fact that $\Delta(j_0^3S)=0$ together with one of the conditions $S_{3,0}S_{1,2}-S^2_{2,1}=0, S_{3,0}\neq 0$, or $S_{0,3}S_{2,1}-S^2_{1,2}=0, S_{0,3}\neq 0$, imply $\kappa(p,q)=0$.
\end{proof}

\begin{rem} Simply saying that $(p,q)\in L$ is a parabolic point of nondegenerate ellipse is not enough to characterize the type of singularity of the Wigner caustic on shell at $(p,q)$. Therefore, for a Lagrangian surface, the type of singularity of the Wigner caustic on shell at a parabolic point of nondegenerate ellipse is a further affine-symplectic invariant that allows for a finer classification of the point. \end{rem}


\begin{thebibliography}{99}



\bibitem{AGV} V. I. Arnol'd, S. M. Gusein-Zade, A. N. Varchenko,
\emph{ Singularities of Differentiable Maps}, Vol. 1, Birhauser,
Boston, 1985.

\bibitem{BM} P. H.~Baptistelli, M. G.~Manoel, The
classification of reversible-equivariant steady-state bifurcations
on self-dual spaces, {\it Math. Proc. Cambridge Philos. Soc.} {\bf 145} (2)
(2008) 379-401.

\bibitem{Ber} M. V. Berry, \emph{Semi-classical mechanics in phase space: A study of Wigner's function},
Philos. Trans. R. Soc. Lond., A 287(1977), 237-271.

\bibitem{BGT} J. W. Bruce, P. J. Giblin, F. Tari, \emph{Families of surfaces: height functions, Gauss maps and duals}.

\bibitem{BKdP} J. W. Bruce, N. P. Kirk, A. A. du Plessis, \emph{Complete transversals and the classification of singularities}, Nonlinearity 10 (1997), 253-275.

\bibitem{Damon2} J. Damon, \emph{The unfoldings and determinacy theorems for subgroups of $\mathcal A$ and $\mathcal K$}, Mem. Am. Math. Soc. 306 (1984).

\bibitem{DRs} W. Domitrz, P. de M. Rios, \emph{Singularities of equidistants and Global Centre Symmetry sets of Lagrangian submanifolds}, arXiv 1007.1939.

\bibitem {GWZ} P. J. Giblin, J. P. Warder and V.M.Zakalyukin, \emph{Bifurcations of affine equidistants}, Proceedings of the Steklov Institute of Mathematics 267 (2009), 57--75.

\bibitem{Gib2} P. Giblin, \emph{Affinely invariant symmetry sets}, Geometry and Topology of Caustics -- Caustics '06. Banach Center Publications Vol 82 (2008), 71-84.



\bibitem{GJ} P. Giblin, S. Janeczko, \emph{Geometry of curves and surfaces through the contact map},  Topology and its Applications 159 (2012), 466-475.



\bibitem{K} N. P. Kirk \emph{Computational aspects of classifying singularities}, LMS J. Comput. Math. 3 (2000) 207-228.

\bibitem{MRR} D. K. H. Mochida, M. C. Romero Fuster, M. A. S. Ruas, \emph{Geometry of surfaces in 4-space from a contact viewpoint}, Geom. Dedicata 54 (1995) 323-332.

\bibitem{OH} A. M. Ozorio de Almeida, J. Hannay, \emph{Geometry of Two Dimensional Tori in Phase Space:
Projections, Sections and the Wigner Function}, Annals of Physics,
138(1982), 115-154.

\bibitem{Poi} H. Poincar\'e, \emph{Les M\'ethodes Nouvelles de la M\'echanique C\'eleste}, vol. 3, Gauthier-Villars, Paris, 1892.

\bibitem{RO1} P. de M. Rios, A. Ozorio de Almeida, \emph{On the propagation of semiclassical Wigner functions},  J. Phys. A: Math. Gen. 35 (2002) 2609-2617.

\bibitem{RO2} P. de M. Rios, A. Ozorio de Almeida, \emph{A variational principle for actions on symmetric symplectic spaces}, J. Geom. Phys. 51, No. 4, 404-441 (2004).


\bibitem{Rob1} M. Roberts, \emph{On the genericity of some properties of equivariant map germs}, J. Lond. Math. Soc., II. Ser. 32(1985), 177-192.

\bibitem{Rob2} M. Roberts, \emph{Characterizations of finitely determined equivariant map germs}, Math. Ann. 275(1986), 583-597.



\end{thebibliography}
\end{document}